\documentclass[11pt]{article}

\usepackage{amssymb,amsthm,amsmath,amssymb,wrapfig,dsfont,authblk}
\usepackage{thmtools,thm-restate}
\PassOptionsToPackage{hyphens}{url}
\usepackage[square,sort,semicolon,numbers]{natbib}

\usepackage{varioref}
\usepackage{verbatim}
\usepackage[usenames,svgnames,xcdraw,table]{xcolor}
\definecolor{DarkBlue}{rgb}{0.1,0.1,0.5}
\definecolor{DarkGreen}{rgb}{0.1,0.5,0.1}
\usepackage[backref=page]{hyperref}
\hypersetup{
     colorlinks   = true,
     linkcolor    = DarkBlue, 
     urlcolor     = DarkBlue, 
	 citecolor    = DarkGreen 
}
\renewcommand*{\backref}[1]{}
\renewcommand*{\backrefalt}[4]{%
    \ifcase #1 (Not cited.)%
    \or        (Cited on page~#2)%
    \else      (Cited on pages~#2)%
    \fi}
\usepackage[capitalise,noabbrev]{cleveref}
\usepackage[margin=0.9in]{geometry}

\usepackage[linesnumbered,lined,boxed,ruled]{algorithm2e}
\SetKwInOut{Parameter}{Parameter}
\SetKwRepeat{Do}{do}{while}
\SetKwFor{ForEach}{for each}{do}{endfor}

\usepackage{nicefrac}

\usepackage{enumerate}
\usepackage[labelfont={normalfont,bf},textfont=it]{caption}
\usepackage{subcaption}
\usepackage[T1]{fontenc}
\usepackage{mathtools}

\usepackage{bbm}

\usepackage{tabulary}
\usepackage{multirow}
\usepackage{booktabs}
\usepackage{colortbl}

\makeatletter
\renewcommand{\paragraph}{%
  \@startsection{paragraph}{4}%
  {\z@}{1.0ex \@plus 1ex \@minus .2ex}{-1em}%
  {\normalfont\normalsize\bfseries}%
}
\makeatother

\usepackage[backgroundcolor=blue!20!white]{todonotes}

\newtheorem{theorem}{Theorem}
\newtheorem{lemma}{Lemma}
\newtheorem{claim}{Claim}

\newtheorem{definition}{Definition}
\theoremstyle{definition}

\allowdisplaybreaks

\DeclareMathOperator*{\argmax}{arg\,max}

\newcommand{\ALG}{\textsc{Alg}}
\newcommand{\MMS}{\mathrm{MMS}}
\newcommand{\EFone}{\mathrm{EF1}}
\newcommand{\NSW}{\mathrm{NSW}}
\newcommand{\EFX}{\mathrm{EFX}}
\newcommand{\A}{\mathcal{A}}
\newcommand{\I}{\mathcal{I}}
\newcommand{\bmu}{\overline{\mu}}

\let\oldnl\nl
\newcommand{\nonl}{\renewcommand{\nl}{\let\nl\oldnl}}%

\begin{document}
\title{\bfseries Fair Division of Indivisible Goods Among Strategic Agents}  
\author{Siddharth Barman\thanks{Indian Institute of Science. \texttt{barman@iisc.ac.in} \\ \hspace*{13pt} The first author was supported in part by a Ramanujan Fellowship (SERB - {SB/S2/RJN-128/2015}).} \quad Ganesh Ghalme\thanks{Indian Institute of Science.  \texttt{ganeshg@iisc.ac.in}} \quad Shweta Jain\thanks{Indian Institute of Science. \texttt{shwetajains20@gmail.com}} \\ \quad Pooja Kulkarni\thanks{Indian Institute of Science. \texttt{poojak@iisc.ac.in}} \quad Shivika Narang\thanks{Indian Institute of Science. \texttt{shivika@iisc.ac.in}}}
\date{}
\maketitle

\begin{abstract}
We study fair division of indivisible goods in a single-parameter environment. In particular, we develop truthful social welfare maximizing  mechanisms for \emph{fairly} allocating indivisible goods. Our fairness guarantees are in terms of solution concepts which are tailored to address allocation of indivisible goods and, hence, provide an appropriate framework for fair division of goods. This work specifically considers fairness in terms of \emph{envy freeness up to one good} ($\EFone$),  \emph{maximin share guarantee} ($\MMS$), and \emph{Nash social welfare} ($\NSW$).

Our first result shows that (in a single-parameter environment) the problem of maximizing welfare, subject to the constraint that the allocation of the indivisible goods is $\EFone$, admits a polynomial-time, $1/2$-approximate, truthful auction. We further prove that this problem is {\rm NP}-Hard and, hence, an approximation is warranted. This hardness result also complements prior works which show that an \emph{arbitrary} $\EFone$ allocation can be computed efficiently.

We also establish a bi-criteria approximation guarantee for the problem of maximizing social welfare under $\MMS$ constraints. In particular, we develop a truthful auction which efficiently finds an allocation wherein each agent gets a bundle of value at least $\left(1/2 - \varepsilon \right)$ times her maximin share and the welfare of the computed allocation is at least the optimal, here $\varepsilon >0$ is a fixed constant. We complement this result by showing that maximizing welfare is computationally hard even if one aims to only satisfy the $\MMS$ constraint approximately.
Our results for $\EFone$ and $\MMS$ are based on establishing interesting {majorization} inequalities. We also observe that the problem of maximizing $\NSW$ in a single parameter environment admits a truthful polynomial-time approximation scheme. 

\end{abstract}

\section{Introduction}

Fairness is a fundamental consideration in many real-world allocation problems. Course assignment~\cite{OSB10finding} and inventory pricing \cite{ROT11} are just two examples of such settings. These two applications, in particular, entail allocation of discrete\footnote{That is, goods which cannot be fractionally assigned to the agents.}  resources and, by contrast,  classical notions of fairness typically address divisible\footnote{Such goods represent resources, such as land, which can be fractionally assigned.} goods~\cite{S80cut,M04fair,BRA95}. While \emph{envy-free}\footnote{An allocation is said to be envy free if, under it, every agent values the bundle assigned to her at least as much as she values any other agent's bundle.} and \emph{proportional}\footnote{An allocation among $n$ agents is said to be proportional if every agent values her bundle at least $1/n$ times the valuation she has for the grand bundle of goods.} allocations always exist for divisible goods, such guarantees do not hold when the goods are indivisible: consider a single discrete good which has to be allocated between two agents. Here, in any allocation, the losing agent will be envious and will not achieve proportionality. 

This gap has been addressed in recent years by efforts that focus on fair division of indivisible goods, see, e.g.,~\cite{LAN16}. This thread of research has lead to an in-depth study of solution concepts which better capture the combinatorial aspects of allocating indivisible goods. 
 Arguably, the three most prominent notions of fairness in this line of work are \emph{envy freeness up to one good} ($\EFone$),  \emph{maximin share guarantee} ($\MMS$), and \emph{Nash social welfare} ($\NSW$). These notions have been well substantiated by the development of complementary existential results, efficient algorithms, and implementations, such as Course Match and Spliddit~\cite{PRO14,B11combinatorial,GHO17,CAR16,BAR18,LMM+04approximately,COL15,OSB10finding,GP15spliddit}. 
 
 However, almost all of this work is confined to settings where the valuations of the agents are known, i.e., to nonstrategic settings.\footnote{Notable exceptions include the work of Amanatidis et al.~\citep{AMA16}, Christoforou et al.~(\citep{CHR16}), and Nguyen et. al. ~\citep{nguyen2013allocation}. A discussion on these results is provided later in this section.} The present paper complements the literature on fair division of indivisible goods by considering this problem in a strategic setting.  

In particular, we focus on settings in which indivisible goods have to be auctioned off among strategic agents/bidders in \emph{single-parameter environments.}  In these well-studied environments, the valuation of each strategic bidder $i$ (over goods) decomposes into the agent's private valuation parameter, $v_i \in \mathbb{R}_+$, and a (global) \emph{public value summarization function} $w: 2^{[m]} \mapsto \mathbb{R}_{+}$, here $m$ is the set of goods, see, e.g.,~\cite{GOE10} and~\cite{Roth17}. Specifically, the valuation of agent $i$ for a subset of goods $S \subseteq {[m]}$ is given by $v_i(S) := v_i w(S)$, where $v_i \in \mathbb{R}_{+}$ is the private parameter of $i$ and $w(\cdot)$ is a summarization function common across all bidders.%

Single-parameter environments are of fundamental importance in mechanism design. Indeed, the foundational result of Myerson~\cite{MYE81} is applicable within this framework and these environments are used to model several important applications, e.g., (i) \emph{Ad auctions}~\cite{EDE05,VAR07} and (ii) \emph{Strategic load balancing}~\cite{ARC01}: (i) Ad auctions are run in real time---every time a user submits a keyword for web search---to assign slots (prominent positions on the webpage) to sponsored advertisements. The auction is used to determine which advertisers' links are displayed, in what order, and the payment charged to them. Here, each slot (i.e., each auctioned good) is associated with a probability with which it will receive a click from the user. These click-through rates correspond to the public value summarization function $w$ and the per-click value of bidder $i$ is modeled by the private valuation  parameter $v_i$. (ii) In the load-balancing context the weight of the loads (e.g., the size of computational jobs) gives us $w$ and the strategic agent's private parameter $v_i$ is the (processing) cost she incurs per unit load. 

Our problem formulations are broadly motivated by the fact that fairness is an important concern in many such applications of single-parameter environments. For example, in ad auctions it is relevant to consider fairness both from a quality-of-service standpoint and for regulatory reasons.\footnote{Imposing fairness has been found to alleviate publicity starvation in ad auctions~\cite{CHR16}. Furthermore, incorporating fairness has the potential of ensuring diversity in ads and, hence, improving end user experience. In addition, fairness guarantees in ad auctions provide a formal way to address mandates (e.g., the European Union competition law) that forbid search engines from implementing monopolistic policies, such as advantageously displaying their own products ~\cite{IND18}. }

In this work we express $\EFone$, $\MMS$, and $\NSW$ in terms of public value summarization function $w$ to obtain distinct formulations for the problem of fairly allocating indivisible goods. These formulations, by construction, provide fairness guarantees which can be independently validated. That is, a bidder can verify that an allocation resulting from the auction is fair even if she is not privy to the payments charged to (and the valuation parameters of) the other bidders. Note that, in this setting, since the valuation of an agent is obtained by scaling the summarization function by an agent's private valuation parameter, the fairness requirements can be equivalently stated in terms of valuations of the agents.

Besides fairness, our objectives conform to the quintessential desiderata of algorithmic mechanism design: we aim to develop computationally efficient, truthful auctions for maximizing social welfare. Our focus on social welfare, in particular, stems from the fact that it is a standard benchmark in auctions; it is often thought that a seller (in competitive settings) should consider maximizing welfare, since otherwise a competitor can potentially steal the customers by doing so.\footnote{Social welfare is an objective of choice in real-world settings as well, e.g., Facebook's online advertising system is based on the VCG mechanism (that maximizes social welfare)~\cite{metzcade2015}.}

\subsection{Additional Related Work}
\label{ssec:Additional}
In contrast to the case of indivisible goods, fair division of \emph{divisible} goods among strategic agents is relatively well-studied; see, e.g., ~\cite{COL13M, CHE13, BEI17, MAY12, branzei2017nash}. In particular, Cole et al.~\cite{COL13M} consider proportional fairness and develop a mechanism without money to allocate divisible items. 

Mechanisms without money have also been developed for indivisible goods by Amanatidis et al.~\cite{AMA16} and Christoforou et al.~\cite{CHR16}. At a high level, the goal of these results is to truthfully elicit the valuations of the agents and, thereby, achieve (approximate) fairness. Unlike the setup considered in ~\cite{AMA16} and~\cite{CHR16}, our focus is on auctions, where payments (money) provide a natural means to achieve truthfulness. Furthermore, our underlying objective is to maximize social welfare, with fairness as a constraint. ~\cite{nguyen2013allocation} also consider fair division of indivisible goods, but establish envy freeness in expectation.




Fairness in the context of auctions has been studied in prior work as well; see, e.g.,~\cite{COH11, GOL03, HAR11, FEL15, GUR05,  TAN15}. These results focus on \emph{envy-free pricing}, i.e., they establish fairness in terms of utilities, not valuations. Furthermore, in this line of work, envy-free pricing is essentially considered as a surrogate for truthfulness. In particular, Goldberg and Hartline \cite{GOL03} prove that no truthful auction that achieves a constant fraction of the optimal revenue is envy free (with respect to the utilities). Cohen et al.~\cite{COH11} show that envy free pricing and truthfulness can be achieved together only if the agents have a homogeneous capacity on the number of goods they receive.  Tang and Zhang~\cite{TAN15} consider envy-free pricing, in lieu of truthfulness, and develop a polynomial-time approximation scheme for maximizing social welfare in sponsored search. As mentioned previously, all the auctions developed in this work are truthful and we study fairness in terms of valuations.\footnote{This leads to publicly-verifiable fairness guarantees.}

\noindent
\subsection{\bf Our Contributions and Techniques}
As mentioned previously, three central solution concepts in the context of fair division of indivisible goods are $\EFone$, $\MMS$ and $\NSW$. Next, we detail our fair auctions (FA) for these three notions. 

\noindent
(i) \textbf{FA-EF1}:  Envy freeness up to one good ($\EFone$) was defined by Budish~\cite{B11combinatorial}. This comparative notion of fairness provides a cogent analogue of envy freeness for the indivisible case. An allocation is said to be $\EFone$ iff, under it, every agent values her bundle at least as much as any other agent's bundle, up to the removal of the most valuable good from the other agent's bundle. Interestingly, an $\EFone$ allocation always exists\footnote{For indivisible goods, such a universal existence guarantee does not hold with respect to envy freeness.} and can be computed efficiently, even under general, combinatorial valuations~\cite{LMM+04approximately}. Another attractive aspect of this fairness notion is that it does not compromise economic efficiency: under additive valuations, there always exists an allocation which is both $\EFone$ and \emph{Pareto efficient}~\cite{CAR16}.\footnote{In the current setup, since the valuations are scalings of the additive public value summarization function, the allocations determined by our mechanisms will always be Pareto efficient.}

We consider the standard objective of maximizing social welfare subject to the $\EFone$ constraints.
This problem is {\rm NP}-hard (Appendix~\ref{sec:EF1Supp}). Complementing this hardness result, we establish an approximation guarantee for the corresponding mechanism design problem. Specifically, we show that the problem of maximizing social welfare, under $\EFone$ constraints, admits a $1/2$-approximate, truthful auction  (Theorem~\ref{thm:efone}). We also prove that in general single-parameter environments\footnote{In general single-parameter environments the public value summarization functions can be agent specific.} (therefore, in multi-parameter environments) a nontrivial approximation guarantee cannot be achieved, under standard complexity-theoretic assumptions (Appendix~\ref{sec:EF1Supp}). This strong negative result is in sharp contrast to the constant-factor approximation guarantee which we obtain for single-parameter environments with identical public valuations.

\noindent 
(ii) \textbf{FA-MMS}: Maximin share guarantee ($\MMS$) is a threshold-based notion defined by Budish \cite{B11combinatorial}. This notion deems an allocation to be fair iff every agent gets a bundle of value at least as much as an agent-specific fairness threshold called the {maximin share}. These shares correspond to the maximum value that an agent can guarantee for herself if she were to (hypothetically) partition the goods into $n$ subsets and, then, from them receive the minimum valued one; here $n$ is the total number of agents. In other words, the maximin share is the value obtained by an execution of the \emph{cut-and-choose} protocol over indivisible goods: an agent forms an $n$-partition of the goods and the remaining $(n-1)$ agents select a subset from the partition before the agent. Hence, a risk-averse agent will form a partition which maximizes the minimum value over the subsets in it.
Our goal is to develop a truthful social-welfare maximizing auction subject to the constraint that each agent receives a bundle of value at least her maximin share. 

As computing the maximin share is {\rm NP}-hard,\footnote{A reduction from the partition problem proves this claim, even for additive and identical valuations.} this paper considers a bi-criteria approximation guarantee. We show that under standard complexity-theoretic assumptions a bi-criteria approximation is unavoidable (Appendix~\ref{sec:AMMSSupp}). 
We develop a truthful auction which efficiently computes an allocation where each agent gets a bundle of value at least $(1/2 -\varepsilon)$ times her maximin share and the social welfare of the computed allocation is at least as much as the optimal (Theorem~\ref{thm:mms}); here $\varepsilon > 0$ is a fixed constant.\footnote{Note that under additive (but nonidentical) valuations an $\MMS$ allocation might not exist \cite{PRO14,KPW16can}.}


\noindent
(iii) \textbf{FA-NSW}:  The Nash social welfare ($\NSW$) of an allocation is defined as the geometric mean of the agents' valuations for their bundles~\cite{N50bargaining,KN79nash}. $\NSW$ provides a measure to quantify the extent of fairness of an allocation~\cite{M04fair,KEL97}. Specifically, for divisible goods, it is known to satisfy strong fairness and efficiency properties~\cite{V74equity}. Even in the indivisible case, if the valuations are additive, then an allocation which maximizes $\NSW$ is both fair ($\EFone$) and Pareto efficient~\cite{CAR16}. 

Though, finding an allocation (of indivisible goods) which maximizes $\NSW$ under additive valuations is {\rm APX}-hard~\cite{Lee17APX}, a number of constant-factor approximation algorithms have been developed for this problem under additive valuations:~\cite{COL15} established a $2.89$ approximation for maximizing $\NSW$, and this approximation ratio has been improved to $e$~\cite{AGS+17nash}, $2$~\cite{COL17} and, most recently, to $1.45$~\cite{BAR17}. In fact, if the valuations of the nonstrategic agents for the indivisible goods are identical and additive, then a polynomial-time approximation scheme (PTAS) is known for maximizing $\NSW$~\cite{nguyen2014minimizing}.

We consider the problem of maximizing Nash social welfare in single parameter environment; see \textsc{FA-NSW} in Section~\ref{sec:prelims}. Specifically, we observe that, for this problem, the approximation result from the nonstrategic setting directly leads to a truthful PTAS (Theorem~\ref{thm:nsw}).  \\

In general, EF1, MMS, and approximate NSW are incomparable notions. ~\cite{CAR16} show that MMS does not imply EF1, and vice versa. Also, simple examples establish that an approximately Nash optimal allocation is not guaranteed to be EF1 or MMS. 

 Our approximation results for all the three notions are based on finding partitions of the goods  considering only their public value $w$. That is, we compute partitions without considering the agents' private valuation parameters and hence the computation is \emph{bid-oblivious}. We show that as long as the $j$th highest (with respect to the public value) subset in the computed partition is allocated to the $j$th highest bidder, the stated approximation guarantees are achieved. In addition, such a sorted allocation leads to a \emph{monotone allocation rule} and provides a truthful auction.

Our algorithms for the $\EFone$ and $\MMS$ formulations are completely combinatorial and can be implemented in sorting time. Given that in many applications the underlying auction has to be executed in real time and at scale, the simplicity and extreme efficiency of the resulting mechanisms are notable merits. The approximation results for the $\EFone$ and $\MMS$ formulations rely on proving \emph{majorization inequalities}; see Definition~\ref{def:major}. In particular, for $\EFone$ we show that all $\EFone$ partitions approximately majorize each other (Lemma~\ref{lemma:major}). This property is interesting in its own right and shows that---independent of the valuation parameters/bids per se---as long as we allocate the partition in a sorted manner the  approximation guarantee holds. 
For the $\MMS$ problem, we design an efficient algorithm which finds a  $\left(1/2 - \varepsilon \right)$-approximate $\MMS$ allocation which majorizes an optimal allocation. 

For the $\NSW$ problem, one can use the fact that the valuations are identical, up to scaling. Specifically, an allocation which (approximately) maximizes $\NSW$ with respect to the public valuation $w$ also (approximately) maximizes $\NSW$ with respect to the valuations. Therefore, using the PTAS of~\cite{nguyen2014minimizing} for maximizing $\NSW$ under additive, identical valuations, we obtain the desired truthful auction. Details of this $\NSW$ result are deferred to a full version of this work and in the remainder of the paper we focus on fairness in terms of $\EFone$ and $\MMS$.

\section{ Preliminaries }
\label{sec:prelims}
We denote an instance of the fair-auction setting $\I$ with $n$ bidders, $[n]=\{1,2,\ldots,n\}$, and $m$ indivisible goods, $[m]=\{1,2, \ldots, m\}$ by a tuple $\langle [m], [n], w, (v_i)_{i\in [n]} \rangle $. The private preference of each agent $i$ is represented by a single parameter $v_i \in \mathbb{R}_+$. In addition, the weight/quality of a subset of goods $S \subseteq [m]$, is specified through a publicly-known summarization function, $w: 2^{[m]} \mapsto \mathbb{R}_+$. The valuation of bidder $i \in [n]$ for a subset of goods $S$ is defined to be $v_i \ w(S)$. Throughout, we will consider $w$ to be additive, i.e., $w(S) := \sum_{g \in S} w(g)$, where $w(g)$ denotes the weight/quality of good $g \in [m]$. For ease of presentation, we overload notation and use $v_i \in \mathbb{R}_+$ to denote the valuation parameter and $v_i( \cdot)$ to denote the valuation function of bidder $i$, $v_i(S) = v_i w(S)$. Furthermore, in this single-parameter environment, since $v_i$s directly scale $w$, the fairness guarantees can be equivalently stated in terms of agents' valuations.

\subsection{Fairness Notions } Write $\Pi_n([m])$ to denote the set of $n$-partitions of the set $[m]$. An  \emph{allocation}  $\A $ $=(A_1, A_2, \ldots, A_n)$ $ \in \Pi_n([m])$ refers to an $n$-partition of $[m]$ in which subset  $A_i$ is assigned to agent $i$. The fairness notions considered in this work are: $\EFone$ and $\MMS$. 

\noindent
(i) An allocation $\A $ is said to be $\EFone$ if for every pair of agents $i, j\in [n]$ there exists a good $g \in A_j$ such that  $v_i(A_i) \geq v_i(A_j) - v_i(g)$. Since $v_i \in \mathbb{R}_+$ is a positive scaling term, the inequality can be rewritten as $w(A_i) \geq w(A_j \setminus g)$ .

\noindent
(ii) Given a fair division instance $\I= \langle [m], [n], w, (v_i)_i \rangle$ the maximin share, $\mu$, is defined as

$$\mu :=  \max_{(P_1,\ldots, P_n) \in \Pi_n([m])} \ \min_{j \in [n]} w(P_j).$$ An allocation $ (A_1, \ldots, A_n)$ is said to be $\MMS$ iff $w(A_i) \geq \mu$ for all agents $i \in [n]$. We can define the maximin share of an agent $i$ as $\MMS_i := \max_{(P_1,\ldots, P_n) \in \Pi_n([m])} \ \min_{j \in [n]} \ v_i (P_j) $. Note that $\MMS_i = v_i \mu$ and we get that an allocation is $\MMS$ iff each agent $i$ receives a bundle of value at least $v_i \mu$. We will also consider allocations which satisfy the $\MMS$ requirement approximately: for $\alpha \in (0,1]$, an allocation $(A_1,\ldots, A_n)$, which satisfies $w(A_i) \geq \alpha \mu$ for all $i \in [n]$, is said to be $\alpha$-approximate $\MMS$.  

\noindent
(iii) The Nash social welfare of an allocation $\A=(A_1, \ldots, A_n) $ is defined as the geometric mean of the agents' valuations, $\left( \prod \limits_{i=1}^n \left(v_i w(A_i) \right) \right)^{1/n}  = \left( \prod \limits_{i=1}^n v_i  \right)^{1/n}  \ \left( \prod \limits_{i=1}^n w(A_i)  \right)^{1/n} $. Note that, for fixed positive scalars  $v_i$s, an allocation that maximizes $\NSW(\A) : = \left( \prod_{i=1}^n w(A_i)  \right)^{1/n}$ also maximizes the Nash social welfare with respect to valuations. 

With these solution concepts in hand, this work considers the standard objective of maximizing welfare. Formally, given instance $\I= \langle [m], [n], w, (v_i)_i \rangle$, the respective optimization problems addressed by the auctioneer are 

{
\begin{figure*}[ht!]
\centering
\fbox{\begin{minipage}[2cm][3.5cm][t]{0.2 \columnwidth}
\centering 
\underline{\textsc{FA-EF1}}
\begin{align*}
\max_{\substack{(S_1,\ldots, S_n)  \\ \in \Pi_n([m])} } & \ \ \ \sum \limits_{i=1}^{n} v_i w(S_i) \\
\text{s.t. }   w(S_i) & \geq w (S_j) - w(g) \\     \text{ for all }  i,j \in & [n]  \text{ some } g \in S_j    
 \end{align*}
\end{minipage} }
\hfill 
\fbox{\begin{minipage}[2 cm][3.5cm][t]{0.2 \columnwidth}

\centering 
\underline{\textsc{FA-MMS}}
\begin{align*}
\max_{\substack{(S_1,\ldots, S_n) \\ \in \Pi_n([m])}} & \ \ \ \sum \limits_{i=1}^{n} v_i w(S_i) \\
\text{s.t. }   w(S_i) & \geq \mu \ \ \ \text{for all } i \in [n]  \\ 
  \ \  ( \mu   \text{ is $\MMS$}& \text{ value under $w$)} 
 \end{align*} 
 \normalsize
\end{minipage}}
\hfill 
\fbox{\begin{minipage}[2cm][3.5cm][t]{0.2 \columnwidth}
\centering 
\underline{\textsc{FA-NSW}}
\begin{align*}
\small
\max_{\substack{(S_1,\ldots, S_n)  \\ \in \Pi_n([m])}}  \ \  \left( \prod_{i=1}^n v_i  \right)^{\frac{1}{n}}  \ \left( \prod_{i=1}^n w(S_i)  \right)^{\frac{1}{n}}   
 \end{align*}
 \normalsize
\end{minipage} }
\label{fig:fig1}
\end{figure*}
}

\subsection{Mechanism Design Terminology} This work develops a single-round, sealed-bid auction for the above mentioned single-parameter environment. Here, the true per-unit valuation, $v_i$, is the only strategic parameter of each agent $i$, which she reports as a bid $b_i \in \mathbb{R}_+$. Hence, the auctioneer has to design a mechanism for the fair-auction instance  $\I= \langle [m], [n], w, (b_i)_{i\in [n]} \rangle$. In particular, the auctioneer requires an allocation rule, $\mathcal{A}$, which maps the submitted bids, $b_i$s, to an allocation, $\mathcal{A}: \mathbb{R}_+^n \mapsto \Pi_n([m])$, along with a payment rule, $p$, which charges a payment of $p_i$ to agent $i$.

In the presence of strategic agents, who could potentially misrepresent the parameter $v_i$, we would like to design allocation rule $\mathcal{A}$ coupled  with an appropriate payment rule $p$ such that truthful reporting is a dominant strategy. That is, the desiderata is to develop a mechanism, $(\mathcal{A}, p)$, which is \emph{dominant strategy incentive compatible (DSIC)} in addition to being fair, welfare maximizing, and computationally efficient. Recall that DSIC provides strong incentive guarantees and it imposes nominal behavioral assumptions. This property requires that for each agent, truthfully reporting her actual valuation is a dominant strategy and it never leads to negative utility. 


For allocation rule $\mathcal{A}(b_1, \ldots, b_n) \in \Pi_n([m])$, write $A_i(b_1, \ldots, b_n)$ to denote the bundle allocated to agent $i$. We consider the standard  quasilinear-utility model wherein utility of agent $i$ under allocation rule $\mathcal{A}$ is given by $u_i(\mathcal{A}(b_i, b_{-i},{w}); v_i):= v_i \ w(A_i (b_i, b_{-i})) - p_i$, here $b_i$ is a bid of agent $i$ and $b_{-i} \in \mathbb{R}_{+}^{n-1}$ are the bids of all the other agents. In the current context, a mechanism $(\mathcal{A} , p ) $ is DSIC iff  $ u_i(\mathcal{A}(v_i, b_{-i}); v_i )$ $ \geq u_i(\mathcal{S}(b_i, b_{-i}); v_i) $ for all bids $ b_i \in \mathbb{R}_{+}$ and reports $b_{-i} \in \mathbb{R}_{+}^{n-1}$. \\

\section{Main Results}
\label{sec:results}
 In this work we establish the following key results 

\begin{restatable}{theorem}{EFONE} 
\label{thm:efone}
\label{THM:EFONE}
There exists a polytime, DSIC mechanism that achieves an approximation ratio of $1/2$ for \textsc{FA-EF1}.
\end{restatable}

 It is relevant to note that \textsc{FA-EF1} is $\rm{NP}$-hard (Appendix~\ref{sec:EF1Supp}). We further complement the approximation guarantee of Theorem~\ref{thm:efone} by showing that it is $\rm{NP}$-hard to obtain an $m^\delta$-approximation for the analogous problem (of maximizing social welfare subject to $\EFone$ constraints) in \emph{general} single-parameter environments; here $\delta>0$ is a fixed constant (Appendix ~\ref{sec:EF1Supp}). This hardness of approximation result is obtained via an approximation-preserving reduction from the Maximum Independent Set problem (Appendix ~\ref{ssec:SSEF1GenHardness}).


In the context of \textsc{FA-MMS}, it is relevant to note that computing the maximin share, $\mu$, is $\rm{NP}$-hard. Hence, we consider a bi-criteria approximation guarantee and establish the following result in Section~\ref{sec:SSMMSProof}.

\begin{restatable}{theorem}{LMMS}
\label{thm:mms}
\label{THM:MMS}
There exists a polynomial time, DSIC mechanism which computes a $\left(1/2 - \varepsilon\right)$-approximate $\MMS$ allocation with social welfare at least as much as the optimal value of \textsc{FA-MMS}, here $\varepsilon \in (0,1)$ is a fixed constant. 
\end{restatable}


In addition, we show that, under standard complexity-theoretic assumptions, an efficient algorithm which achieves  a nontrivial approximation for \textsc{FA-MMS}, without violating the $\MMS$ constrains, does not exist; see Appendix~\ref{sec:AMMSSupp}. 

In Section~\ref{sec:SS-NSW} we observe that a truthful PTAS can be developed for the third formulation.
\begin{restatable}{theorem}{LNSW}
\label{thm:nsw}
\label{THM:NSW}
There exists a DSIC mechanism which, for every fixed $\varepsilon >0$, finds a $(1+\varepsilon)$-approximate solution of \textsc{FA-NSW} in polynomial time.
\end{restatable}






\section{ Mechanism Design to Algorithm Design } 
\label{sec:mdtoad}
An auction  $(\mathcal{A},p)$ is given by an allocation rule $\mathcal{A}: \mathbb{R}^n_+ \mapsto \Pi_n([m])$ which maps the bids, $(b_i)_{i\in [n]}$, to a partition of goods, and a payment rule, $p$, which specifies the payment $p_i$ charged to agent $i \in [n]$. We rely on the foundational result of Myerson (Lemma~\ref{lemma:myerson}) which asserts that for DSIC mechanisms it suffices to construct \emph{monotone} allocation rules. 

\begin{definition}[Monotone Allocation Rule] 
An allocation rule $\mathcal{A}: \mathbb{R}^n_+ \mapsto \Pi_n([m])$ is said to be monotone if for all agents $i \in [n]$ and bids of other agents $b_{-i} \in \mathbb{R}_+^{n-1}$, the allocation to agent $i$ is non-decreasing, i.e., $w(A_i(z, b_{-i}))$ is non-decreasing function of bid $z \in \mathbb{R}_+$. Here, $w$ is the public value summarization function and $A_i(z, b_{-i})$ is the bundle allocated to $i$ under the bid profile $(z, b_{-i})$.
\end{definition}

\begin{lemma}[Myerson's Lemma~\cite{MYE81}]
\label{lemma:myerson}
Under a single parameter, quasilinear utility model, an allocation rule, $\mathcal{A}$, can be coupled with a payment rule, $p$, to obtain a DSIC mechanism, $(\mathcal{A},p)$, if and only if $\mathcal{A}$ is monotone. Further, the payment rule, $p$, which renders such an allocation rule truthful (DSIC) is unique and is given by 
$$p_i(b_i, b_{-i}) := b_i w(A_i(b_i, b_{-i})) - \int \limits_{0}^{b_i} w(A_i(z, b_{-i})) dz $$ . 
\end{lemma}
 The payment formula simplifies further, since the underlying allocation rule is piecewise constant: $p_i(b_i, b_{-i}) =$ $  \sum_{j=1}^{l} z_j. \text{jump in } w(A_i( \cdot , b_{-i})) \text{ at } z_j$. Here, $z_j$s are the bid values at which the allocated bundle $A_i( . , b_{-i})$ changes.
 
Throughout, we will develop monotone allocations by first computing a partition,\footnote{A collection of pairwise-disjoint subsets whose union is $[m]$.} $\{P_i\}_{i\in [n]}$, of the indivisible goods $[m]$ and then allocating the $j$th highest (with respect to $w(\cdot)$) subset from the partition to the $j$th highest bidder. Formally, we call such allocations sorted. 
\begin{definition}[Sorted Allocation]
Given a partition $ \mathcal{P} = \{P_i\}_{i\in [n]} \in \Pi_n([m])$ of the indivisible goods $[m]$, which satisfies $w(P_1) \geq w(P_2) \geq \ldots \geq w(P_n)$, we say that a corresponding allocation, $\mathcal{S} =(S_1, \ldots, S_n) \in \Pi_n([m])$, is \emph{sorted} with respect to the given bids $(b_i)_{i\in[n]}$ iff $S_{\sigma(j)} = P_j$, for all $j \in [n]$, where $\sigma$ is the order in which the bids are sorted: $b_{\sigma(1)} \geq b_{\sigma(2)} \geq \ldots \geq b_{\sigma(n)}$.
\end{definition}

For the two problems \textsc{FA-EF1} and \textsc{FA-MMS} we will develop algorithms which will be \emph{bid oblivious}.\footnote{As mentioned previously, an analogous result for \textsc{FA-NSW} follows directly from the fact that the $\NSW$ objective is scale invariant.} In particular, the developed algorithms will find a partition $\{P_i\}_{i\in [n]}$ of the $m$ indivisible goods by only considering the function $w(\cdot)$, and not the bids $b_i$s. The input to all our algorithms will be $\langle [m], [n], w \rangle$; this will explicitly ensure that they are bid oblivious. 
Since the fairness notions are completely specified in terms of $w(\cdot)$, it is intuitive that a fair partition can be computed without considering the bids. The notable property of our algorithms is that they find a partition $\mathcal{P}$ which provides a ``universal'' approximation guarantee: as long as we perform a sorted allocation of $\mathcal{P}$ the stated approximation ratio is achieved, independent of the bids per se. That is, if the $j$th highest (with respect to $w(\cdot)$) subset of $\mathcal{P}$ is allocated to the $j$th highest bidder, then the stated approximation is obtained.   


For the two problems we first compute a partition $\mathcal{P}$, using a bid oblivious algorithm, and, then, perform a sorted allocation of $\mathcal{P}$. This ensures that the resulting allocation rule not only satisfies the desired approximation guarantee, but is also monotone. Therefore, via Myerson's Lemma, we obtain a DSIC mechanism. These observations imply that the underlying mechanism design problem reduces to developing bid-oblivious algorithms which find fair partitions with above-stated   universal approximation guarantee. The rest of the paper primarily addresses this algorithm design goal.

\section{ Proof of Theorem~\ref{thm:efone}}
\label{sec:SSEF1Proof}
This section presents a $1/2$-approximation algorithm for \textsc{FA-EF1}. Recall that in this problem $\EFone$ is defined with respect to a single, additive function $w(\cdot)$. Hence, if the subsets $P_i$s of a partition $\mathcal{P} = \{P_1, \ldots, P_n\}$ satisfy the $\EFone$ condition under $w$, any allocation of $\mathcal{P}$ (independent of which agent gets which subset $P_i$) will be $\EFone$. 

Our bid-oblivious approximation algorithm is simple: return \emph{any} partition $\mathcal{P}= \{P_1,\ldots, P_n\}$ which satisfies the $\EFone$ condition under $w$; such a partition can be computed in sorting-time by the round-robin algorithm~\cite{CAR16}. As mentioned above, any allocation of $\mathcal{P}$---in particular, the sorted allocation---is $\EFone$. This directly takes care of the constraints in \textsc{FA-EF1}. To address the objective, we prove that such a sorted allocation (independent of our choice of the partition $\mathcal{P}$ which satisfies $\EFone$) has social welfare at least half of the optimal. The approximation guarantee follows from an interesting result (Lemma~\ref{lemma:major}) which shows that all the partitions which are $\EFone$ (under an identical, additive function $w$) \emph{$1/2$-majorize} each other. 

 

\begin{definition}[$\beta$-Majorization:]
\label{def:major}
A sequence $(x_i)_{i=1}^{n}$ is said to $\beta \in \mathbb{R}_+$ majorize another sequence $ (y_i)_{i=1}^{n}$ iff 
\begin{align*}
\sum \limits_{i=1}^{k} x_{(i)}  & \geq \beta \sum \limits_{i=1}^{k} y_{(i)}  \qquad \ \text{ for all } 1 \leq k \leq n-1 \text{  and } \\ 
\sum \limits_{i=1}^{n} x_i & = \sum \limits_{i=1}^{n} y_i.
\end{align*}
\end{definition}

The following ``universal'' approximation guarantee holds for $\beta$-majorizing sequences (Appendix~\ref{sec:EF1Supp}): if a sequence $(x_i)_{i}$ $\beta$-majorizes another sequence $ (y_i)_{i}$, then, for any set of parameters $v_1 \geq v_2 \geq \ldots \geq v_n \geq 0$, we have 
$\sum_i v_i x_{(i)} \geq \beta \sum_i v_i y_{(i)} $. 
This guarantee, along with Lemma~\ref{lemma:major} (which shows that $1/2$ majorization holds between any two $\EFone$ partitions), imply that the sorted allocation of an arbitrary $\EFone$ partition is a $1/2$-approximate solution to \textsc{FA-EF1}.\footnote{This $1/2$-approximation guarantee  also extend to $\EFX$  \cite{CAR16} constraints. An allocation  $\mathcal{P}\in \Pi_n([m])$ is said to be $\EFX$ if $v_i(P_i) \geq v_i(P_j) - v_i(g) $ for all $i,j \in [n]$ and for all $ g \in P_j$. Note that an $\EFX$ allocation always exists in the current  setting \cite{PLA18} }

\begin{lemma}
\label{lemma:major}
Let $\mathcal{P} = (P_i)_{i=1}^{n}$ and $\mathcal{Q} = (Q_i)_{i=1}^{n}$ be two partitions (of the $[m]$ goods) which satisfy the $\EFone$ condition under the additive function $w$. Then, the sequence $(w(P_i))_i$ $1/2$-majorizes  the sequence $(w(Q_i))_i$.
\end{lemma}
\begin{proof}
Reindexing does not affect the majorizing order. Hence, without loss of generality, we assume that $w(P_1)\geq w(P_2) \geq \ldots \geq w(P_n)$ and $w(Q_1)  \geq \ldots \geq w(Q_n)$. Also, we consider the partitions to be distinct, otherwise, the claim holds trivially. Write index $\hat{i} :=\max \{ i \in [n]  \mid w(P_i) > w(Q_i)\}$; such an index $\hat{i}$ exists since the partitions are distinct and $ \sum_{i=1}^n w(P_i) = \sum_{i=1}^n w(Q_i)$.
For each index $\ell \in [n-1]$ we establish the required inequality, $\sum_{i=1}^\ell w(P_i) \geq \frac{1}{2} \sum_{i=1}^\ell w(Q_i)$, by considering two cases:

\noindent {\textbf{Case 1:}}  $\ell \leq \hat{i}$: Since $\mathcal{Q} = (Q_i)_i$ is an $\EFone$ partition, $w(Q_k) \leq w(Q_{\hat{i}})+ w(g_k)$, for all  $ k \in [n]$, where $g_k \in \arg \max \limits_{g\in Q_k} w(g) $. Therefore,
 
\begin{align}
 \sum_{k=1}^{\ell} w(Q_k) \leq \ell w(Q_{\hat{i}})+ \sum_{k=1}^\ell w(g_k) \leq \ell w(P_{\hat{i}}) + \sum_{k=1}^\ell w(g_k) \label{eqn:one}  
\end{align}
Write $g_{(1)},g_{(2)},\ldots,g_{(\ell)}$ to denote the $\ell$ highest (with respect to $w(\cdot)$) goods in $[m]$ and note that  $\sum_{k=1}^\ell w(g_k) \leq \sum_{k=1}^\ell w(g_{(k)})$. Let $P_{(1)}, P_{(2)}, \ldots, P_{(t)}$ denote the minimum collection of subsets in the partition whose union contains these goods $H:=\{ g_{(1)},g_{(2)},\ldots, g_{(\ell)}\}$; here, $t$, the number of subsets required to cover $H$ is at most $\ell$. Also, note that $P_{(i)}$s are pairwise disjoint and each good, $g_{(j)}$ belongs to exactly one of them. In addition, since $P_i$s are indexed such that $w(P_1) \geq w(P_2) \geq \ldots \geq w(P_n)$, we have 
\begin{align}
\sum \limits_{k=1}^t w(P_{(k)})  \leq  \sum \limits_{k=1}^t w(P_k) \leq \sum \limits_{k=1}^\ell w(P_k) .
\label{eq:two}
\end{align}

Therefore,
\begin{align*}
\sum \limits _{k=1}^\ell w(Q_k)  &\leq \ell w(P_{\hat{i}}) + \sum \limits_{k=1}^\ell w(g_k) \leq \ell w(P_{\hat{i}}) + \sum \limits_{k=1}^t w(P_{(k)}) \tag{using Eq (\ref{eqn:one}) and definition of $P_{(i)}$s} \\ & \leq \ell w(P_{\hat{i}}) + \sum \limits_{k=1}^\ell w(P_{k})  \tag{using (\ref{eq:two})}.
\end{align*}

\normalsize
The fact that $w(P_{\hat{i}}) \leq w(P_{k})$ for all $ k \leq \ell \leq \hat{i}$ implies $\sum_{k=1}^\ell w(Q_k) \leq 2 \sum_{k=1}^\ell w(P_{k})$ for all $\ell \leq \hat{i}$. Hence, $1/2$-majorization holds for all $\ell \leq \hat{i}$. 

\noindent {\textbf{Case 2:}}  $\ell >\hat{i}$: Note that $\sum_{k=1}^n w(P_k)=\sum_{k=1}^n w(Q_k)$. Splitting the sums gives us $\sum_{k=1}^\ell w(P_k) + \sum_{k=\ell+1}^n w(P_k)=\sum_{k=1}^\ell w(Q_k) + \sum_{k=\ell +1}^n w(Q_k)$.  That is, 
$\sum_{k=1}^\ell w(P_k)  =\sum_{k=1}^\ell w(Q_k) + \sum _{k=\ell +1}^n (w(Q_k) - w(P_k))$.
By definition of $\hat{i}$, for all indices $k \geq \ell > \hat{i}$, we have $w(Q_k) \geq w(P_k)$. Therefore, the second summand on the right-hand side of the previous equation is non-negative. This implies $\sum_{k=1}^\ell w(P_k) \geq \sum_{k=1}^\ell w(Q_k)$.
This inequality, along with the non-negativity of $w$, directly establish $1/2$-majorization for all indices $\ell \geq \hat{i}$: $2 \sum_{k=1}^\ell w(P_k) \geq \sum_{k=1}^\ell w(Q_k)$.
\end{proof}


\noindent  From Lemma~\ref{lemma:major} and the universal approximation guarantee of majorizing sequences we get that the social welfare from any two sorted $\textsc{EF1}$ allocations   $\mathcal{P}$ and  $\mathcal{Q}$ is related by  $SW(\mathcal{P}) \geq \frac{1}{2} SW(\mathcal{Q})$. 
This completes the proof of Theorem~\ref{thm:efone} , which we restate below.

\EFONE*

\section{Proof of Theorem~\ref{thm:mms}}
\label{sec:SSMMSProof}

This section addresses \textsc{FA-MMS}. Since computing  $\mu$ (the maximin share) is {\rm NP}-hard, 
we use the PTAS of~\cite{WOE97} to compute an estimate $\bmu$ in polynomial time which satisfies $(1-\varepsilon) \mu \leq \bmu \leq \mu$, for fixed constant $\varepsilon \in (0, 1/2)$. 
Write allocation $\mathcal{O} = (O_i)_{i=1}^n \in \Pi_n([m])$ to denote an optimal solution of \textsc{FA-MMS}. Note that the feasibility of $\mathcal{O}$ ensures $w(O_i) \geq \mu $ for all $i \in [n]$. Let {\rm SW}$(\mathcal{O})$ denote the  social welfare of $\mathcal{O}$. Throughout this section, by reindexing, we assume that the agents are indexed in non-increasing order of their valuations, i.e., $v_1 \ge v_2 \ge \ldots \ge v_n \geq 0$. 
 
Our algorithm, $\ALG$, finds a partition $\mathcal{P} = (P_i)_{i=1}^n$ which satisfies $w(P_i) \geq \frac{1}{2} \bmu \geq \left( \frac{1 - \varepsilon}{2} \right) \mu$.  We will show that a sorted allocation of $\mathcal{P}$, say $\mathcal{A} = (A_i)_{i=1}^n$, achieves social welfare at least as much as {\rm SW}$(\mathcal{O})$. Note that $\mathcal{A}$ will be a $\left( \frac{1 - \varepsilon}{2} \right)$-approximate $\MMS$ allocation. Given that the monotone allocation rule returns sorted allocation $\mathcal{A}$, we get a DSIC mechanism that satisfies Theorem~\ref{thm:mms}.



$\ALG$ partitions the goods in two sets based on their $w$-value: high-valued goods $\overline{H}:= \{ g \in [m] \mid w(g) \ge \overline{\mu}\}$ and low-valued goods $\overline{L}:= \{ g \in [m] \mid w(g) < \overline{\mu}\}$. It further divides $\overline{L}$, into two disjoint subsets, tiny goods, $\overline{L}_1 := \{ g \in [m] \mid w(g) < \frac{\overline{\mu}}{2} \}$ and medium-valued goods, $\overline{L} \setminus \overline{L}_1$. Clearly, the set $\overline{L} \setminus \overline{L}_1$ contains all the goods that are neither in $\overline{H}$ nor in $\overline{L}_1$. For analysis, we also consider two sets $H := \{g \mid w(g) \ge \mu \}$ and $L := \{g \mid w(g) < \mu \}$. Since $\bmu \leq \mu$, we have $\overline{L} \subseteq L$.

We will primarily focus on the case wherein $\textsc{Alg}$ returns the partition $\mathcal{P}$ after Phase 4 (Step 19).\footnote{ This implies $\tau < n$.} The analysis of the remaining cases (in which either \textsc{Alg} terminates after Phase 1 ($\tau^{'} \geq n $) or \textsc{Alg} terminates after Phase 2 ($\tau \geq n$)) is fairly similar.  

We now show in Lemma~\ref{lemma:validmms} that $\mathcal{P}$ is approximately fair.   
 {

\begin{algorithm}[ht!]
    \SetKwInOut{Input}{Input}
    \SetKwInOut{Output}{Output}
    \SetAlgoNoLine
    \renewcommand{\thealgocf}{}
    \DontPrintSemicolon
    \SetAlgorithmName{$\ALG$}{ }{ }
    \caption{Bi-Criteria Approximation Algorithm for \textsc{FA-MMS} }
        \label{alg:OptMMS}
    \Input{Additive $w:2^{[m]} \mapsto \mathbb{R}_+$ and estimate $\overline{\mu}$} 
    \Output{Approx.~$\MMS$ partition: $ \mathcal{P}$ = $\{P_1, \ldots,P_n\} $}
    \nonl           --------\texttt{Phase 0: Initialization}-------- \;
                $P_i = \emptyset$  for all $ i \in [n],     \overline{H} = \{g \in [m] \mid w(g) \geq \overline{\mu} \},   \overline{L} =[m] \setminus \overline{H},\overline{L}_1 = \{ g \in \overline{L} \mid w(g) < \frac{\overline{\mu}}{2}\}$  \;

     \nonl      --------\texttt{Phase 1: Medium Valued Goods}-------- \;
    For all the goods  $g  \in \overline{L} \backslash \overline{L}_1 $, form singleton subsets $P_1, P_2 , \ldots P_{\tau^{'}}$\;
    \If{$\tau^{'} \geq n $}{ Return the partition $\{ P_1, P_2,  \ldots, P_{n-1}, [m]\setminus \cup_{i\in [n-1]} P_i \}$}
      \nonl             --------\texttt{Phase 2: Tiny Valued Goods}-------- \;
        Initialize index $a = \tau^{'} +1$  \;
    \While{$\overline{L}_1 \neq \emptyset$}{
    \While { $w(P_a) < \frac{\overline{\mu}}{2}$ }{
Pick a good $g \in \overline{L}_1$ and update $  P_a \leftarrow P_a \cup \{g\}$ and $ \overline{L}_1 \leftarrow \overline{L}_1 \setminus \{g \}$\;
      }
      $a \leftarrow a+1$ \;
      }
       \If{for the last formed subset $P_a$ we have $w(P_a) < \frac{\overline{\mu}}{2}$ }{
            $E \leftarrow P_a$  \hspace{5pt} \tcp{\textsc{ set aside leftover goods  }}
             Set $P_a = \emptyset$,     and        $a \leftarrow a -1$
 \hspace{5pt} \tcp{\textsc{ empty  $P_a$ }}
}
            Let $\tau$ be the total number of nonempty subsets, $\{P_1, P_2, \ldots, P_\tau \}$, populated so far \;
                  \If{$\tau \geq  n $}{
      Return the partition $\{ P_1, P_2,  \ldots, P_{n-1}, [m]\setminus \cup_{i\in [n-1]} P_i \}$ }

           \nonl        --------\texttt{Phase 3: High Valued Goods}-------- \;
              
      Sort $\overline{H}$ in a non-decreasing order, i.e., $ \overline{\mu} \leq w(g_1^{'}) \leq w(g_2^{'}) \leq \ldots , \leq  w(g_{|H|}^{'}) $ \;
    Form $(n- \tau -1)$ singleton subsets, i.e., $P_{\tau+i} = \{ g_i^{'}\}$  for $1 \leq i \leq n- \tau -1$\;
        \nonl       --------\texttt{Phase 4: Leftover Goods}-------- \;

     $P_n \leftarrow E \cup \{g_{n-\tau}^{'},\ldots,g_{|H|}^{'}\} $ \;
    Return the partition $\{P_1,P_2,\ldots,P_n\}$
    
  \end{algorithm}
 }
  \normalfont
\begin{lemma}
\label{lemma:validmms}
 The algorithm $\ALG$ computes, in polynomial time, an $n$-partition $\mathcal{P} = (P_i)_{i=1}^n$ which satisfies $w(P_i) \ge \frac{\overline{\mu}}{2}$ for all $i \in [n]$.  
\end{lemma}
First, we note the following useful  property of the optimal allocation $\mathcal{O} =(O_i)_i$: besides $O_1$ (the bundle allocated to the highest-valued agent), any bundle $O_i$ is either completely contained in $L$ or it is disjoint from it. 
\begin{claim}
\label{claim1}
For $i \neq 1$, either $O_i \subseteq L$ or $O_i \subseteq H$.
\end{claim}
\begin{proof}
For contradiction, say there exists a bundle $O_i$ (with $i \neq 1$) such that both the subsets $O_i \cap L$ and $O_i \cap H$ are nonempty. Write $g'$ to be a good in $O_i \cap H$. Then, by swapping all the goods $O_i \setminus \{g' \}$ from agent $i$ to agent 1 increases the social welfare; recall that $v_1 \geq v_i$. Since $w(g') \geq \mu$, even after the swapping agent $i$ is left with the singleton $\{g'\}$, which preserves the feasibility of the resulting allocation, i.e., the $\MMS$ requirement for agent $i$ is maintained. This contradicts the optimality of $\mathcal{O}$. 
\end{proof}

Note that $L$ and $H$ are disjoint. Hence, $O_i \subseteq H$ implies $O_i \cap L = \emptyset$. Let $I^* := \{ i \in [n] \mid O_i \subseteq L \}$ denote the bundles in $\mathcal{O}$ which are completely contained in $L$ . The following claim directly follows from the definition of $I^*$.  
\begin{claim}
\label{claim2}
$\sum_{i \in I^*} w(O_i) \le w(L)$.
\end{claim}

Analogous to $I^*$, we consider subsets in the computed partition which are completely contained in $L$; in particular, write $I := \{ i \in [n]  \mid P_i \neq \emptyset \text{ and } P_i \subseteq L \}$. We will establish that all the subsets in the computed partition are nonempty, hence set $I$ will be the indices of all the subsets which are completely contained in $L$. 
Note that $\ALG$ populates the first $\tau$ (see Step 13) subsets using only low-valued goods (using $\overline{L} \subseteq L$ in Phase 1 and 2), hence, $ [\tau] \subseteq I$. 
The next claim\footnote{In the interest of space, the proof of this claim is omitted.} bounds the weight of each of the first $\tau$ subsets formed by $\ALG$. 
\begin{claim}
\label{claim3}
For all $ i \in [\tau]$, we have $\frac{\overline{\mu}}{2} \le w(P_i) \le \mu$.
\end{claim}

Besides the leftover goods ($E$ in Step 11), $\overline{L}$ is entirely distributed among $\{P_1, P_2, \ldots, P_\tau \}$. This leads to the following upper bound. 
\begin{claim}
\label{claim4}
$w(\overline{L}) \le \sum_{i=1}^\tau w(P_i) + \frac{\mu}{2}$.
\end{claim}
\begin{proof}
We first note that $\overline{L} = E \cup (P_i)_{i=1}^{\tau}$. This implies $w(\overline{L}) = \sum_{i = 1}^{\tau} w(P_i) + w(E) \leq \sum_{i = 1}^{\tau} w(P_i) + \mu/2 $. The last inequality follows from the fact that $\overline{\mu} \leq \mu$. 
\end{proof}

Recall that $I^*$ and $I$ denote the indices of bundles in $\mathcal{O}$ and $\mathcal{P}$ respectively, which are completely contained in $L$.
We now compare the cardinalities of these sets. 
\begin{claim}
\label{claim5}
$| I| \geq |I^*|$.
\end{claim}
\begin{proof}
$\ALG$ populates the first $\tau$ subsets from $\overline{L}$ and goods from the set $L \setminus \overline{L}$ are assigned as singletons in Phase 3, or allocated to $P_n$; recall $\mu \geq \bmu$ and $\overline{L} \subseteq L$. We also have, for all $i \in I \setminus \{ n \}$,  $w(P_i) \leq \mu$; in particular, for the first $\tau$ subsets via Claim 3. 

First, consider the case wherein $|I^*| = n$, i.e., $L = [m]$. Here, accounting for the cumulative weight, $w([m])$, shows that $|I| =n$ as well. Specifically, the cumulative weight is at least $n \mu$, since $w([m]) = \sum_{i=1}^n w(O_i) \geq n \mu$. Using the fact that $w(P_i) \leq \mu$ for all $i \in I \setminus \{ n \}$, we get that $\ALG$ will necessarily construct $n$ nonempty subsets. Thus, in this case, $|I| = n$.

We now address the case where $|I^*| \leq n-1$. We consider two (sub) cases, either (i) $|L \setminus \overline{L}| >  n- \tau - 1$, or (ii) $|L \setminus \overline{L}| \leq n - \tau -1$. The sorting in Step 16 of $\ALG$ ensures that the goods from $L \setminus \overline{L}$ are assigned (as singletons) before any good from $H$. Therefore under (i), $\ALG$ is, in fact, able to construct at least $n-1$ subsets from $L$ (i.e., $|I| \geq n-1$): $\tau$ from $\overline{L}$ and $n - \tau - 1 $ (in Phase 3) from $L \setminus \overline{L}$. This gives us $|I| \geq |I^*|$. 
Under (ii), all the $\ell:= | L \setminus \overline{L}|$ goods are assigned as singletons in Phase 3. Hence, $ |I| = \tau + \ell$. 
Using the previous claims, we will show that $\tau + \ell \geq |I^*|$. 
\begin{align*}
|I^*|  \mu & \leq \sum_{i \in I^*} w(O_i) \tag{since $w(O_i) \ge \mu$ for all $i \in [n]$} \\
& \le w(L) \tag{via Claim \ref{claim2}} \\ 
& \le w(\overline{L}) + \ell \cdot \mu \tag{since $w(g) \le \mu$ for all $g\in L \setminus \overline{L}$}  \\
 &\le \sum_{i =1}^\tau w(P_i) + \frac{\mu}{2} + \ell \cdot \mu \tag{via Claim \ref{claim3}} \\
 & \le \left(\tau + \ell + \frac{1}{2}\right) \mu \tag{via Claim \ref{claim4}}
\end{align*}
Since, $|I^*|$, $\tau$, and $\ell$ are integers, we get $\tau + \ell \geq |I^*|$. 
\end{proof}

\begin{proof}[{\bf Proof of Lemma~\ref{lemma:validmms}}]
$\ALG$ ensures that for each $i \in I$ we have $w(P_i) \geq \bmu/2$: for all the subsets in $[\tau] \subseteq I$, this inequality holds via Claim~\ref{claim3}. The remaining subsets, $P_i$ with $i \in I$ (constructed in Phase 3 of $\ALG$), are singletons consisting of goods (from the set $L \setminus \overline{L}$) each of weight at least $\bmu$ and, hence, the inequality holds for them as well. 
Therefore, to complete the proof it suffices to show that $|H| \geq n- |I|$. Since, this would guarantee that (in Phase 3 and 4) each of the last $n - |I|$ subsets get a good from $H$ each of value $\bmu$ (recall, $w(g) \geq \mu \geq \bmu$ for all $g \in H$). That is, every constructed subset $P_i$ satisfies $w(P_i) \geq \bmu/2$.

The definition of $I^*$ implies that for all $j \in [n] \setminus I^*$ we have $O_j \not\subseteq L$, i.e., $O_j \cap H \neq \emptyset$ for all $j \in [n] \setminus I^*$.\footnote{Without loss of generality, we can assume that the subsets $O_j$s are nonempty.} Hence, Claim~\ref{claim5} gives us $|H| \geq n - |I^*| \geq n - |I|$. This establishes the stated bound for all the computed subsets $P_i$s. 
\end{proof}
Overall, we get that $\ALG$, in polynomial time, computes a partition $\mathcal{P}$ which satisfies the  $(1-\varepsilon)/2$-approximate $\MMS$ guarantee.  Next, we show that social welfare of   a sorted allocation of $\mathcal{P}$,  is at least as much as the social welfare of the optimal solution of \textsc{FA-MMS} problem. 
\begin{lemma}
Let $\mathcal{P}$ be a partition computed by $\ALG$ and $\mathcal{A}=(A_i)_{i\in [n]}$ be a sorted allocation of $\mathcal{P}$. Then, the sequence $(w(A_i))_{i \in [n]}$ majorizes $(w(O_i))_{i \in [n]}$, where $\mathcal{O} = (O_i)_{i\in [n]}$ is an optimal solution of \textsc{FA-MMS}.
\label{lemma:optimality}
\end{lemma}
\noindent The following claims are  used in establishing Lemma~\ref{lemma:optimality}.
\begin{claim}
\label{claim:sort}
Let $\mathcal{P} =\{P_1, \ldots, P_n \}$ be a partition returned by $\ALG$. Then, $w(P_n) \geq w(P_i)$ for all $i \in [n]$. 
\end{claim}
\begin{proof}
If $H \neq \emptyset$, then the sorting step in Phase 3 (Step 16) guarantees that $P_n$ receives the highest valued good(s) in $H$. For every other subset $P_i$ in $\mathcal{P}$, either $P_i$ is a singleton (with a lower-valued good from $H \cup (L \setminus \overline{L})$) or $P_i \subseteq \overline{L}$ (i.e., $i \in [\tau]$), in which case $w(P_i) \leq \mu$ (Claim~\ref{claim3}). Therefore, if $H$ is nonempty, the stated claim follows.
On the other hand, if $H = \emptyset$, i.e., $L = [m]$, $\ALG$ will allocate goods so that $w(P_i) \leq \mu$ for all $i \in [n-1]$. Given that the cumulative weight is at least $n \mu$, again we have $w(P_n) \geq \mu \geq w(P_i)$ for all $i$. 
\end{proof}

\begin{claim}
\label{claim6}
$w(A_j) \le w(O_j)$ for all $j \in \{2, 3, \ldots, n \}.$
\end{claim}
\begin{proof}
Let $F := \{ i \neq 1 \mid A_i \subset L \}$ denote the set of agents, besides agent 1, who get only the low valued goods. Note that the bundles in $F$ are obtained by allocating subsets from $\{P_i\}_{i \in I}\setminus \{P_n\}$; recall that $P_n$ is allocated to agent 1 (Claim ~\ref{claim:sort}). 
We prove this claim by considering indices in $F$ and in $\{2, 3, \ldots, n\} \setminus F$ separately.  

\noindent
(i) If index $j \in F $, then $A_j \in \{P_i\}_{i \in I}\setminus \{P_n\}$. Subsets in this collection receive goods of weight at most $\mu$ during the execution of $\ALG$. Hence, $w(A_j) \leq \mu \leq w(O_j)$. 
 
\noindent
(ii) If index $ j \in \{2, 3, \ldots, n\} \setminus F$, then $A_j \notin \{P_i\}_{i \in I}\setminus \{P_n\}$. Therefore, $A_j$ corresponds to a partition which must have been populated after Phase 1 and 2 of $\ALG$. Here, $j \neq 1$ and, hence, $A_j \neq A_1 = P_n$. This, in turn, implies that the subset corresponding to $A_j$ was formed in Phase 3 and it consists of a single good from the set $H$, i.e., $w(A_j) \geq \mu$. In other words, in this case we must have $|H| \neq \emptyset$. Write $h:=|\{ f \neq 1 \mid A_f \cap H \neq \emptyset \}|$ to denote the number of bundles in $\mathcal{A}$, besides $A_1$, which contain a good from $H$. These $h$ bundles are singletons (see Phase 3) and satisfy $w(A_f) \geq \mu$. Furthermore, since the subsets in $\{P_i\}_{i \in I}\setminus \{P_n\}$ are of weight at most $\mu$, without loss of generality we have that in the sorted allocation $A_2, A_3, \ldots, A_{h+1}$ are bundles formed from $H$.  Therefore, the index $j$ must satisfy $2 \leq j \leq h+1$. 

Analogous to $h$, we consider $h^*:=|\{ f \neq 1 \mid O_f \cap H \neq \emptyset \}|$. In particular, write $O_{(1)}, O_{(2)}, \ldots O_{(h^*)}$ to denote the $h^*$ optimal bundles (besides $O_1$) which contain a good from $H$. Each $O_{(f)}$ is a singleton (Claim~\ref{claim1}) and we index them such that $w(O_{(1)}) \geq w(O_{(2)}) \geq \ldots \geq w(O_{(h^*)})$. Using Claim~\ref{claim5} and the fact that $\ALG$ allocates goods from $H$ as singletons (besides the ones assigned to $A_1=P_n$) we get $h \leq h^*$.  Furthermore, the sorting in Step 17 ensures that the lowest (in terms of weight) $h$ goods from $H$ are allocated as singletons in Phase 3. Therefore, $w(O_{(i)}) \geq w(A_{i+1})$ for all $ i \in \{1,2, \ldots, h\}$. Since $w(O_{i+1}) \geq w(O_{(i)})$ for all $ i \in \{1,2, \ldots, h\}$ (the optimal bundle is necessarily sorted), we get the desired inequality $w(O_j) \geq w(A_j)$ for the index $j$. 
\end{proof}

\begin{proof}[{\bf Proof of Lemma~\ref{lemma:optimality} }]
The cumulative weight of the goods remains same irrespective of the allocation, $\sum_{i = 1}^{n}w(A_i) = \sum_{i = 1}^{n} w(O_i)$. Hence, 
$\sum_{i = 1}^{k}w(A_i) = \sum_{i = 1}^{k} w(O_i) + \sum_{i = k+1}^{n}w(O_i) - \sum_{i = k+1}^{n}w(A_i)$ for all $k \in \{1,2,\ldots n-1\}$.
Applying Claim~\ref{claim6}, $\sum_{i = k+1}^{n}w(O_i) - \sum_{i = k+1}^{n}w(A_i) \ge 0$. Therefore, the stated inequality holds: 
$\sum_{i = 1}^{k}w(A_i) \ge \sum_{i = 1}^{k} w(O_i)$ for all $k \in \{1,2,\ldots n-1\}$.
\end{proof}

The analysis above focused on the case in which $\textsc{Alg}$ returns the partition $\mathcal{P}$ after Phase 4. In the complementary cases (wherein \textsc{Alg} terminates after Phase 1 or Phase 2) the fairness guarantee (Lemma ~\ref{lemma:validmms}) is directly achieved. In addition, in these cases we have the following bounds: $\frac{\overline{\mu}}{2} \leq w(P_i) \leq \overline{\mu} \leq \mu$  for all $i < n$  and $w(P_n) \geq  \mu + \sum_{i=1}^{n -1} ( \mu -w(P_i)) \geq \mu$ (Claim ~\ref{claim:sort}). Using these facts, one can show, as before, that the majorization guarantee  (Lemma~\ref{lemma:optimality}) continues to hold. Therefore, even if $\textsc{Alg}$ terminates before Phase 4, we have {\rm SW}$(\mathcal{A}) \geq$ {\rm SW}$(\mathcal{O})$.

Hence, we get that the sorted allocation $\mathcal{A}$ assigned to the agents satisfies the required fairness and welfare guarantees. Overall, using this sorted (monotone) allocation rule we get a DSIC mechanism which establishes Theorem~\ref{thm:mms}. 
 
\section{ Proof of Theorem~\ref{thm:nsw}}
\label{sec:SS-NSW}

This section describes a DSIC mechanism for \textsc{FA-NSW}. $\NSW$ is a scale-free objective; in particular, the following equality holds for any $v_1, \ldots, v_n \in \mathbb{R}_+$, $\argmax_{(Q_1, \ldots, Q_n) \in \Pi_n([m])} \left( \prod_{i=1}^n v_i C(Q_i) \right)^{1/n} = \argmax_{(Q_1, \ldots, Q_n) \in \Pi_n([m])} \left( \prod_{i=1}^n  C(Q_i) \right)^{1/n}$.


Hence, an allocation that (approximately) maximizes the Nash social welfare under the additive function $w(\cdot)$ is a (approximate) maximizer of \textsc{FA-NSW}. Constant-factor approximation algorithms are known for maximizing Nash social welfare under additive valuations~\cite{COL15, AGS+17nash, COL17, BAR17}. In fact, Nguyen and Rothe~\cite{nguyen2014minimizing} provide a PTAS for this problem when the valuations of the agents are additive and identical. Hence, in a bid oblivious manner and for any fixed $\varepsilon > 0$, we can find a partition, $\mathcal{P} = \{P_1, \ldots, P_n\}$, in polynomial time which satisfies $ \left( \prod_{i=1}^{n}  w(P_i) \right)^{1/n} \geq \frac{1}{1 + \varepsilon} \max_{(Q_1, \ldots, Q_n) \in \Pi_n([m])} \ \ \left( \prod_{i=1}^{n}  w(Q_i) \right)^{1/n}$. 

A sorted allocation of $\mathcal{P}$ continues to satisfy this approximation guarantee. Hence, the monotone allocation rule described in Section~\ref{sec:mdtoad} leads to a DSIC mechanism which satisfies Theorem~\ref{thm:nsw} (restated next).

\LNSW*

The PTAS in~\cite{nguyen2014minimizing} is not combinatorial--it uses the integer programming algorithm of Lenstra~\cite{lenstra1983integer}. We can, however, obtain a combinatorial (and extremely efficient) mechanism for \textsc{FA-NSW} using the $1.061$-approximation algorithm developed in~\cite{BAR18}. 

Also, note that the Nash social welfare approximation guarantee does not require a sorted allocation of $\mathcal{P}$. This requirement is used to only ensure monotonicity of the allocation rule. In fact, for $\NSW$ under additive valuations, we can achieve a $1.45$ approximation in general single-parameter environments. Specifically, if $w_i:2^{[m]} \mapsto \mathbb{R}_+$ is an additive public value summarization function of agent $i$ and $v_i \in \mathbb{R}_+$ is her valuation parameter, then, by definition, the value that agent $i$ has  a subset of goods $S \subseteq [m]$ is $v_i w_i(S)$.\footnote{In single parameter environment $w_i$ is equal to $w$ for all agents $i \in [n]$.} Hence, using the algorithm of Barman et al.~\cite{BAR17},  we can efficiently find a $1.45$-approximation to a Nash optimal under $w_i$s. A sorted allocation of such an approximate solution leads to the stated guarantee for general single-parameter environments.  




\section{Conclusion and Future Work}

This paper studies fair allocation of indivisible goods among strategic agents. In particular, we develop truthful and efficient mechanisms which are (approximately) fair in terms of $\EFone$, $\MMS$ and optimizing $\NSW$. Under $\EFone$ and $\MMS$, we focus on the fundamental objective of maximizing social welfare and obtain constant-factor approximation guarantees (bi-criteria in the $\MMS$ case). 

Going forward, it would quite interesting to address revenue maximization. Note that in the standard Bayesian framework the \emph{virtual valuations} of the agents can be negative. Hence, it is not clear if bid oblivious algorithms exist when the objective is to maximize  expected revenue: if the virtual valuations of all the agents are negative then an empty allocation is both fair and optimal. On the other hand, if all the virtual valuations are positive then we must assign the entire set of goods to obtain high revenue. This separation highlights the fact that, when maximizing revenue, it is not clear if one can decide on the ``right'' subset of goods to allocate without looking at the bids. 

Improving the approximation factors and considering non-additive public value summarization functions (e.g., submodular) also remain interesting directions for future work. 




\bibliographystyle{alpha}
\bibliography{references.bib}

\appendix

\section{EF1 supplements}

\label{sec:EF1Supp}

\subsection{Hardness of \textsc{FA-EF1}}
\label{ssec:SSEF1Hardness}
This section shows that the decision version of \textsc{FA-EF1}---i.e., the problem of finding a social welfare maximizing EF1 allocation---is {\rm NP}-hard by reducing the partition problem to it. The decision version of \textsc{FA-EF1} and the partition problem are as follows:

Decision version of $\textsc{FA-EF1}$:
Given valuation parameters $v_1, \ldots, v_n$, an additive function $w:2^{[m]} \mapsto \mathbb{R}_+$ over $m$ goods along with a threshold $\tau$, does there exist an $\EFone$ (under $w(\cdot)$) allocation $S= (S_i)_{i=1}^n$ that satisfies $\sum _{i=1} ^n v_i w(S_i)\geq 
\tau$? \\

\noindent 
{Partition Problem:} Given a set of integers $A = \{a_1,\ldots,\, a_m\}$, such that $\sum _{i=1}^m a_i=2B$,  find a $2$-partition $\{P_1, P_2\}$ of $A$ which satisfies $\sum _{a_i\in P_1} a_i=\sum _{a_i\in P_2} a_i=B$.

\begin{theorem}
\textsc{FA-EF1} is {\rm NP}-hard.
\end{theorem}

\begin{proof}
Given an instance of the  partition problem, $A= \{a_1, a_2, \ldots, a_m\}$, we construct an instance of \textsc{FA-EF1} with two agents and $(m+1)$ goods. With a positive integer $v$, we set the valuation parameters of the two agents  as  $v_1=v$ and $v_2=v-\epsilon$. The additive function $w$ is set to satisfy $w(i) = a_i$ for all $1\leq i \leq m$ and $W({m+1}) = B+ \epsilon $, with  $\epsilon >0$ and $B+\epsilon > \max _{i=1,\cdots, m} a_i$. Furthermore, let the threshold be $\tau := 3Bv +\epsilon v - \epsilon B$. 

Next, we show that a balanced partition $\{P_1, P_2\}$ exists if and only if, in the constructed instance there exists an $\EFone$ allocation with social welfare at least $\tau$. 

For the forward direction, say in the given partition instance a $2$-partition $\{P_1,P_2\}$ satisfies $\sum _{a_i\in P_1} a_i=\sum _{a_i\in P_2} a_i=B$. Note that the allocation $\mathcal{S} := (S_1,S_2)$, with the first agent's bundle $S_1 := P_1 \cup \{m+1\}$ and the second agent's bundle $S_2 := P_2$ is  $\EFone$. In addition the social welfare of this allocation is equal to $\tau$. 
\begin{align*}
 \sum_{i=1}^{2} v_i w(S_i) &= v(2B +\epsilon)+(v-\epsilon)B \\
&=3Bv +\epsilon v - \epsilon B \\ & = \tau
\end{align*}

We now prove the other direction using a contrapositive argument. In particular, we assume that the given partition instance does not admit a balanced partition.

Say, $P_1$ and $P_2$ are subsets such that $(P_1 \cup\{m+1\}, P_2)$ is an $\EFone$ allocation. Note that $\{P_1, P_2\}$ corresponds to a $2$-partition of $A$. Hence, by the assumption that the underlying partition instance does not admit a balanced partition, we get $w(P_1) \neq w(P_2)$. Write $w(P_1) = \Delta$ and $w(P_2) = 2B - \Delta$. The $\EFone$ condition ensures that $\Delta < B$; equality does not hold since  $w(P_1) \neq w(P_2)$. 

There are two possible allocations, either $P_1 \cup\{m+1\}$ is allocated to the first agent or to the second. We show that in both of these cases the social welfare is strictly less than $\tau$. 

If $P_1 \cup\{m+1\} $ is allocated to agent one, then the social welfare is 
\begin{align*}
&v(B+\epsilon + \Delta) + (v-\epsilon) (2B - \Delta )  \\
&=vB +v\epsilon + v\Delta +2Bv -2B\epsilon  -\Delta v +\Delta \epsilon \\
&=3Bv +\epsilon v -\epsilon B -\epsilon(B -\Delta) \\
&<\tau
\end{align*}

Otherwise, if agent two gets $P_1 \cup\{m+1\} $, the social welfare is 
\begin{align*}
&v(2B - \Delta) + (B+ \epsilon +\Delta) (v-\epsilon) \\
&=2vB -v\Delta + vB +\epsilon v +v\Delta -B\epsilon -\epsilon ^2 - \Delta \epsilon \\
&=3Bv +\epsilon v -\epsilon B -\epsilon(\epsilon +\Delta) \\
&<\tau
\end{align*}

Hence, in both the cases the $\EFone$ allocation fails to achieve the threshold welfare $\tau$. This completes the proof. 
\end{proof}

\subsection{Hardness of Maximizing Social Welfare Under $\EFone$ Constraints and Heterogeneous Valuations}
\label{ssec:SSEF1GenHardness}
This section considers $\EFone$ in a setting wherein the valuations of the $n$ agents over the $m$ goods are not necessarily identical. In particular, we show that under additive valuations, $v_i: 2^{[m]} \mapsto \mathbb{R}_+$, for agent $i \in [n]$, the following problem does not admit a nontrivial approximation.


\begin{align*}
\max_{(S_1,\ldots, S_n) \in \Pi_n([m])} &  \ \ \ \sum_{i=1}^{n} v_i (S_i) \\
\text{s.t. }   v_i(S_i) & \geq v_i (S_j) - v_i(g) \,\,\, \text{for all }  i,j \in [n]  \text{ and some } g \in S_j 
\end{align*}

We call this problem \textsc{HET-EF1} and establish an approximation preserving reduction to it from the maximum independent set problem. 

Given a graph $G=(V,E)$ we construct a \textsc{HET-EF1} instance such that the graph has an independent set of size at least $k$ if and only if the constructed \textsc{HET-EF1} instance achieves social welfare at least $k-1$. Set the number of goods $m=|V|$ and the number of agents $n=|E|+1$. For every edge $e_i=(u_j, u_k) 
\in E$, we have an agent $i$ such that $v_i(u_j)=v_i(u_k)=\varepsilon=O(1/n^2)$ and $v_i(u)=0$ for all $ u \notin \{u_j, u_k\}$. Agent $0$ values each good at one, i.e., $v_0(u)=1$ for all  $u \in [m]$.


Let $\mathcal{S}= (S_i)_{i=0}^n$ denote an optimal solution of  the constructed \textsc{HET-EF1} instance. The following lemma shows that the cardinality of the bundle allocated to agent $0$, i.e., $|S_0|$, dictates the social welfare of $\mathcal{S}$. 

\begin{claim}
\label{claim:one}
The social welfare of an optimal allocation $\mathcal{S} = (S_i)_{i=0}^n$, ${\rm SW}(\mathcal{S})$, satisfies $|S_0| \leq {\rm SW}(\mathcal{S}) < |S_0| +1$. 
\end{claim}

\begin{proof}
By definition, ${\rm SW}(\mathcal{S}) = \sum \limits_{i=0}^{n} v_i (S_i)=|S_0|+\sum \limits_{i=1}^{n} v_i (S_i)\geq |S_0|$.

Furthermore, since $v_i(S_i) \leq 2\epsilon$, for all $ 1 \leq i \leq n$, and $\epsilon = O(1/n^2)$, we have ${\rm SW}(\mathcal{S}) = |S_0|+\sum \limits_{i=1}^{n} v_i (S_i)\leq |S_0| +\epsilon|E| \leq |S_0| +1/n <|S_0|+1$.
\end{proof}

\begin{claim}
\label{claim:two}
The bundle allocated to agent $0$ (i.e., $S_0$) corresponds to an independent set in the given graph $G$. 
\end{claim} \label{thm:het}

\begin{proof}
For contradiction, say  $u_j,u_k\in S_0$ for an edge $e_i = (u_j,u_k)$ of the graph. 

In such a case,  $v_i(S_i)=0$ and $v_i(S_0)=2\epsilon$. In addition, even after a good ($u_j$ or $u_k$) is removed from $S_0$, agent $i$ will continue to envy agent $0$, which contradicts the fact that $(S_i)_{i=0}^n$ is an $\EFone$ allocation. Therefore, $S_0$ is an independent set. 
\end{proof}

\begin{lemma}
In the given graph $G$, the maximum independent set is of size $t$ if and only if the optimal social welfare in the constructed \textsc{HET-EF1} instance is between $t$ and $t+1$.
\end{lemma}

\begin{proof}
If $G$ has an independent set of size $t$, then allocating the corresponding set of goods to agent $0$, we get that the social welfare is at least $t$. Note that such an allocation is $\EFone$.

Furthermore, if the optimal social welfare is equal to $t+\delta$, where $0\leq \delta <\epsilon|E|<1$, then Claim~\ref{claim:one} implies that agent $0$ is allocated $t$ goods. Therefore, using Claim~\ref{claim:two}, we have that $G$ has an independent set of size at least $t$. 
\end{proof}

The previous lemma shows that there exists an approximation preserving reduction from the maximum independent set problem to \textsc{HET-EF1}. Since maximum independent set cannot be efficiently approximated within $|V|^{\delta}$ (assuming ${\rm NP} \not\subseteq {\rm ZPP}$), for a constant $\delta >0$~\cite{HAS96}, we get the following theorem.   

\begin{theorem}
For a fixed constant $\delta>0$, there does not exist a polynomial-time $m^{\delta}$-approximation algorithm for \textsc{HET-EF1}, unless ${\rm NP} \subseteq {\rm ZPP}$.
\end{theorem}

Note that an $m$-approximation for \textsc{HET-EF1} can be obtained by first finding a maximum weight matching between the agents and the goods, and then---with this partial allocation in hand---executing the algorithm of Lipton et al.~\cite{LMM+04approximately}.

\subsection{Approximation Guarantee for $\beta$-Majorizing Sequences}
\label{app:major}

The following lemma is well known under the majorization order (i.e., for $\beta = 1$); see, e.g.,\citep{MAR11}. 
Here, for completeness, we provide a proof for general $\beta \in (0,1]$.  

\begin{lemma} 
\label{lemma:betamaj}
If sequence $(x_i)_{i}$ $\beta$-majorizes sequence $ (y_i)_{i}$, then for any set of valuation parameters $v_1 \geq v_2 \geq \ldots \geq v_n \geq 0$ the following inequality holds $\sum_i v_i x_{(i)} \geq \beta \sum_i v_i y_{(i)}$. Here $x_{(i)}$ and $y_{(i)}$ denote the $i$th largest element in the two sequences, respectively, and parameter $\beta \in (0,1]$.
\end{lemma}
\begin{proof}
Reindex the sequences such that $x_1 \geq x_2 \geq \ldots \geq x_n$ and $y_1 \geq y_2 \geq \ldots \geq y_n$ and, hence, the following inequality is satisfied for all $k \in [n-1]$
\begin{align}
\sum \limits _{i=1}^k x_i\geq \beta \sum \limits _{i=1}^k y_i \label{eq:one}
\end{align}

Note that, for any $k \in [n]$ 
\begin{align*}
\sum \limits _{i=1}^k v_i x_i  & = \sum \limits _{\ell=1}^{k-1} \sum \limits _{i=1}^\ell (v_\ell - v_{\ell+1}) x_i \ + \  v_k \sum \limits _{i=1}^k x_i \\ 
& =\sum \limits _{\ell=1}^{k-1} ( v_\ell - v_{\ell +1}) \sum \limits _{i=1}^\ell x_i  \ + \ v_k \sum \limits _{i=1}^k x_i \\
& \geq \sum \limits _{\ell=1}^{k-1} ( v_\ell - v_{\ell +1}) \left( \beta \sum \limits _{i=1}^\ell y_i \right)  \ + \ v_k \left( \beta \sum \limits _{i=1}^k y_i \right) 
\quad \text{(using (\ref{eq:one}) and the ordering $v_\ell \geq v_{\ell+1}$)} \\
& =  \beta \sum \limits _{\ell=1}^{k-1} (v_\ell-v_{\ell+1})\sum \limits _{i=1}^\ell y_i + \beta v_k \sum \limits _{i=1}^k y_i \\
& = \beta \sum \limits _{i=1}^k v_i y_i
\end{align*}

Hence, the claim follows. 


\end{proof}

\section{MMS Supplements}
\label{sec:AMMSSupp}
This section establishes the hardness of maximizing social welfare subject to approximate maximin constraints. 
Given a fair division instance $\mathcal{I} = \langle [m], [n], w, (v_i)_i \rangle$, with additive function $w(\cdot)$, computing the maximin share, $\mu := \max_{(P_1, \ldots, P_n)}$ $ \min_{j} w(P_j)$, is {\rm NP}-hard. Thus, finding a feasible solution of \textsc{FA-MMS} is computationally hard as well. However, using the PTAS developed in~\cite{WOE97}, for any fixed $\varepsilon >0$ and additive function $w$, one can find a $(1 - \varepsilon)$-approximate $\MMS$ allocation in polynomial time. 

In light of this PTAS (and other approximation results~\cite{AMN+15approximation,PRO14,BAR17a,GHO17}), one must address the complexity of a relaxation of \textsc{FA-MMS} where the allocations are required to be approximately maximin fair. We show even this version of the problem is computationally hard, i.e., for a fixed $\alpha \in (0,1)$, we prove the following is {\rm NP}-hard.
\begin{align*}
\max_{(S_1, \ldots, S_n) \in \Pi_n([m])} \quad & \ \ \sum_{i=1}^n v_i w(S_i) \tag{\textsc{FA-AMMS}} \\
\text{subject to} \quad &   w(S_i)  \geq \alpha \mu \qquad \text{ for all } i \in [n].
\end{align*}
\begin{theorem}
For any $\alpha \leq 1/4$, \textsc{FA-AMMS} is {\rm NP}-hard.
\end{theorem}
\begin{proof}
We provide a reduction from the partition problem. Given a set of $m$ integers $W :=\{w_1, w_2, \ldots, w_m\}$, the objective in the partition problem is to determine if there exists a $2$-partition of $[m]$, say $(P_1, P_2)$  such that $\sum_{j \in P_1} w_j = \sum_{j \in P_2} w_j$, i.e., the goal is to find a balanced partition of the given set of integers. 

We will construct an instance of \textsc{FA-AMMS} with $(m+3)$ goods and three agents such that the social welfare of the constructed instance is at least $10 T$ iff the underling partition instance admits a balanced partition; here $T := \frac{1}{2} \sum_{j=1}^m w_j$. In addition, we set the approximation factor $\alpha = 1/4$. This will establish the stated hardness result. 
In particular, we consider a fair-division instance with $v_1 = 1$ and $v_2 = v_3 = 0$ for the three agents. Here, for the first $m$ goods the weight is defined by equating to the given integers:  $ w(j) := w_j$ for all $1 \leq j \leq m$. The weight of the $(m+1)$th and the $(m+2)$th good is set to be $3T$ (recall $T = \frac{1}{2} \sum_{j=1}^m w_j$). Finally, for the last good we have  $w(m+3) := 4T$. 

Note that the maximin share of the constructed instance, $\mu$, is at most $4T$. The total weight $w([m+3])$ is equal to $12T$. Therefore, the average between the three agents is $4T$. Since $\mu$ cannot exceed the average, the stated bound follows. 
Next, we show that $\mu$ is equal to this average iff there is balanced partition of the given set of integers. 

\begin{claim}
The maximin share, $\mu$, of the constructed instance is equal to $4T$ if and only if the underlying partition instance admits a balanced partition.
\end{claim}
\begin{proof}
The reverse direction of this claim is direct. If $(P_1, P_2)$ is a balanced partition then assigning good $(m+3)$ to agent 1 along with the bundles $P_1 \cup \{m+1\}$ and $P_2 \cup \{ m+2\}$ to agents 2 and 3, respectively, ensures that $\mu = 4T$. 

Consider the case in which the partition instance does not admit a balanced partition. Here, we will show that, for any allocation $(A_1, A_2, A_3) \in \Pi_3([m+3])$, the following bound holds $\min_i w(A_i) < 4T$. We can assume that the $(m+3)$th good is allocated as a singleton; otherwise redistributing the goods allocated with this good does not reduce $\mu$ below $4T$. Say, $A_1 = \{m+3\}$. In case the $(m+1)$th and the $(m+2)$th good are allocated together, say the are in bundle $A_2$, then the stated bound holds $w(A_3) \leq w([m]) = 2T < 4 T$. Otherwise, if the $(m+1)$th and the $(m+2)$th good are allocated to, say, agents $2$ and $3$, respectively, we have $\min \{ w(A_2 \setminus \{m+1\}), w(A_3 \setminus \{m+2 \}) \} < T$. The last inequality follows from the fact that the partition instance does not admit a balanced partition. Therefore, we have $\min \{w(A_2), w(A_3) \}  < 4T $. 
\end{proof}

First we will show that if there exists a balanced partition of the given set of integers then the social welfare of the constructed \textsc{FA-AMMS} instance is at least $10T$. For a balanced partition $\{P_1, P_2\}$, we have $w(P_1) = w(P_2) = T$. Therefore, the allocation with bundles $A_1 =\{ m+1, m+2, m+3\}$, $A_2 = P_1$ and $A_3 = P_2$ is $1/4$-approximate $\MMS$; recall $\alpha = 1/4$ and (in this case) $\mu = 4T$. Furthermore, the social welfare of this allocation is equal to $10T$. We complete the proof by establishing that if the given partition instance does not admit a balanced partition then the social welfare is strictly less than $10T$.  Write $\beta := $ $ \max \limits_{(Q_1, Q_2) \in \Pi_2([m])} \min\limits_{i}$ $ w(Q_i)$. The nonexistence of a balanced partition implies that $\beta < T$. 

Write $(P_1, P_2) \in \argmax_{(Q_1, Q_2) \in \Pi_2([m])} \min_{i} w(Q_i)$. Considering the allocation $(\{m+3\}$, $P_1 \cup \{m+1\}$, $P_2 \cup\{ m+2\})$, we get that the maximin share of the constructed instance is at least $\beta + 3T$. Therefore, the approximate maximin share value imposed in \textsc{FA-AMMS} is strictly greater than $\beta $ i.e.
\begin{align*}
\frac{1}{4} \mu \geq \frac{1}{4} \left( \beta + 3 T \right) > \beta. \end{align*}

Note that a feasible allocation for \textsc{FA-AMMS}, say $\mathcal{A} = (A_1, A_2, A_3)$, must satisfy $\min_i w(A_i) \geq \frac{1}{4} \mu > \beta$. This implies that we cannot partition just the first $m$ goods between agents 2 and 3 to obtain a feasible allocation: in particular, if $A_2 \cup A_3 \subseteq [m]$, then $\min \{w(A_2),   w(A_3) \} \leq \beta$ which contradicts the feasibility of $ \mathcal{A}$. Therefore, one of the last three goods, $\{m+1, m+2, m+3\}$ has to be assigned to either agent 2 or agent 3. Say, the $(m+2)$th good is not assigned to agent 1, then we have $w(A_1) \leq w(m+3) + w(m+1) + w([m]) \leq 4T + 3T + 2T = 9T$. Analogously, in the other cases one can argue that $w(A_1) \leq 9T$. Overall, we get that the social welfare is strictly less than $10T$, whenever the given instance does not admit a balanced partition. This completes the proof of the theorem. 
\end{proof}



\end{document}